\documentclass[letterpaper]{article}
\usepackage[utf8]{inputenc}
\usepackage[english]{babel}
\usepackage{graphicx}
\usepackage{subcaption}
\usepackage{tcolorbox}
\usepackage{mathtools}
\usepackage{amsmath}
\usepackage{algorithm}
\usepackage{algcompatible}
\usepackage{hyperref}
\usepackage{algpseudocode}
\usepackage[toc,page]{appendix}
\usepackage{pifont}
\usepackage{amsthm}
\usepackage{caption}
\usepackage{tikz}
\usepackage{tikz-qtree}
\usepackage{upgreek}

\newtheorem{theorem}{Theorem}[section]

\newtheorem{lemma}[theorem]{Lemma}
\theoremstyle{definition}
\newtheorem{definition}{Definition}[section]

\algrenewcommand\textproc{}

\algdef{SE}[SUBALG]{Indent}{EndIndent}{}{\algorithmicend\ }%
\algtext*{Indent}
\algtext*{EndIndent}

\usetikzlibrary{decorations.pathreplacing,shapes.misc,arrows,automata,positioning,shapes,shadows,fit}
\tikzset
  {dest/.style={circle,draw,minimum width=1mm,inner sep=0mm},
   start/.style={diamond,fill,draw,minimum width=1mm, minimum height=1mm,inner sep=0mm},
   terminal/.style={circle,draw,minimum width=1.5mm,inner sep=0mm,label={[label distance=-2.7mm]{\tiny+}}}
  }

\title{Distributed Approximation Algorithms for\\ the Combinatorial
  Motion Planning Problem}

\author{
\normalsize Simran Dokania \\
\normalsize \texttt{sdokania@acm.org}\\
  \and
\normalsize Aditya Paliwal\\
\normalsize \texttt{adityapaliwal@acm.org}\\
  \and
\normalsize Shrisha Rao\\
\normalsize \texttt{shrao@acm.org}\\
}

\date{}

\begin{document}

\maketitle

\begin{abstract}
We present a new $4$-approximation algorithm for the Combinatorial
Motion Planning problem which runs in
$\mathcal{O}(n^2\pmb\upalpha(n^2,n))$ time, where $\pmb\upalpha$ is
the functional inverse of the Ackermann function, and a fully
distributed version for the same in asynchronous message passing
systems, which runs in $\mathcal{O}(n\log_2n)$ time with a message
complexity of $\mathcal{O}(n^2)$. This also includes the first fully
distributed algorithm in asynchronous message passing systems to
perform ``shortcut" operations on paths, a procedure which is
important in approximation algorithms for the vehicle routing problem
and its variants. We also show that our algorithm gives feasible
solutions to the $k$-TSP problem with an approximation factor of $2$
in both centralized and distributed environments. The broad idea of
the algorithm is to distribute the set of vertices into two subsets
and construct paths for each salesman over each of the two
subsets. Finally we combine these pairwise disjoint paths for each
salesman to obtain a set of paths that span the entire graph. This is
similar to the algorithm by Yadlapalli et. al. \cite{3.66} but differs
in respect to the fact that it does not require us to use minimum cost
matching as a subroutine, and hence can be easily distributed.

\end{abstract}

\section{Introduction}

The traveling salesman problem is one of the most famous and one of the most thoroughly studied NP-complete problems in literature. There has been an increasing amount of interest in the various variations of the traveling salesman problem and the vehicle routing problems over the past few decades due to their vast number of practical applications \cite{Application1} \cite{Application5}.

Routing of unmanned vehicles for scouting and reconnaissance operations usually involves using multiple vehicles of different types, i.e., not all vehicles may be equipped with the same sensors and thus some might be unsuitable to scout certain locations. Typically these unmanned vehicles are assigned targets so as to minimize some cost function like fuel consumed, time taken, etc. Such problems can be modelled as various variants of the vehicle routing problem \cite{1.5}, \cite{1.66}, \cite{2_1}, \cite{2_2}.

Robotic arms in assembly lines are typically assigned tasks as a collection of points in 3 dimensional space that they must visit and then perform some operation at these locations. The order in which these locations are visited affects how much energy is consumed by the robot and how much time is taken to complete the entire operation. Usually, these set of target locations can be visited by multiple robots working simultaneously to speed up the process. This problem can also be modelled as an instance of a variant of the vehicle routing problem. Solving such an instance optimally can imply large amounts of savings in time and money for industries \cite{Application2}.

With the boom in e-commerce, logistics is becoming an increasingly important division in companies and efficient and timely delivery of goods is a critical factor to stay alive in these competitive markets. As the number of warehouses and delivery fleets of these e-commerce giants increases, so does the difficulty of the task to optimally route these vehicles to deliver the goods whether to minimize the total time taken or amount of fuel consumed. This again highlights the importance of efficient algorithms to come up with efficient algorithms for the vehicle routing problem and its variants.

In this paper, we will present the first distributed approximation algorithm for a variant of a multiple depot-terminal Hamiltonian path problem called the Combinatorial Motion Planning problem (CMP). Since CMP is known to be NP-Complete (because it is a generalization of the Traveling Salesman Problem which is NP-Complete \cite{Sahni}), we will focus on approximation algorithms, which are polynomial time algorithms but relax the requirement of optimality, i.e., they are algorithms that produce feasible solutions of a given optimization problem within some known bound of error relative to the optimal solution to the problem. A formal definition of an approximation algorithm can be found in appendix A. In this article, we will concern ourselves with approximation algorithms where the approximation factor $\alpha$ is a constant.

Additionally, we would like to point out that there is very little research for distributed approximation algorithms for vehicle routing problems and their variants and it is becoming increasingly important due to the vast number of practical applications, to find efficient techniques in parallelism and approximation algorithms to solve instances of these problems efficiently.

Yadlapalli et. al. \cite{3.66} present a $\frac{11}{3}$ factor approximation algorithm for the Combinatorial Motion Planning problem (CMP). In this article we will present a $4$ factor approximation algorithm for the Combinatorial Motion Planning problem and then subsequently present the first fully distributed algorithm for CMP in asynchronous message passing systems. Though our algorithm gives a worse approximation factor in centralized systems, we show that we would obtain better a approximation factor compared to a distributed version of the algorithm presented by Yadlapalli et al. \cite{3.66}. Specifically, we show that a distributed version of the algorithm by Yadlapalli et al. \cite{3.66} would give an approximation factor of $\frac{13}{3}$ due to lack of exact distributed algorithms for minimum cost matching. Moreover, our algorithm has a significantly better running time of $\mathcal{O}(n^2 \pmb\upalpha(n^2,n))$ compared to a running time of at least $\mathcal{O}(n^3)$ by Yadlapalli et al. \cite{3.66} in centralized systems. Since the inverse of the Ackermann function $\pmb\upalpha$ is an extremely slowly growing function, our algorithm runs in $\mathcal{O}(n^2)$ time for all practical network sizes. In a distributed system, our algorithm runs in $\mathcal{O}(n\log_2{n})$ time and uses at most $\mathcal{O}(n^2)$ messages.

We also present the first fully distributed algorithm for performing the \textit{shortcutting} step in asynchronous message passing systems. The \textit{shortcut} operation is a very common operation used in literature for approximation algorithms for the traveling salesman problem \cite{Christofides} and a fully distributed routine for the same should aid us in developing fully distributed approximation algorithms for the traveling salesman problem. This should help us in addressing the lack of research in producing fully distributed algorithms for the traveling salesman problem and its variants. We believe that such work is necessary since as the size of problem instances gets larger, it will eventually become infeasible to run centralized algorithms to solve these large instances (since these algorithms usually have quadratic or cubic running times) especially in online variants of these problems. We should try to exploit some form of parallelism in distributed systems and develop algorithms that can be run feasibly even for larger instances of the problem.

In this article we also demonstrate a 2-approximation algorithm for the multiple traveling salesman problem ($k$-TSP) by reducing it to an instance of CMP. Subsequently, we present the first distributed approximation for $k$-TSP in asynchronous message passing systems which runs in $O(n\log_2{n})$ time and uses at most $\mathcal{O}(n^2)$ messages.

This article is organized as follows: in Section 2 we formally define the Combinatorial Motion Planning problem and define some notation that will be used throughout this paper along with some preliminaries. In Section 3, we highlight the issues faced by the current state of the art algorithms for CMP in distributed environments and subsequently present our 4-approximation algorithm for CMP in centralized systems. In Section 4 we present a fully distributed version of the $4$-approximation algorithm for CMP. In Section 5 we demonstrate how the multiple Traveling Salesman Problem can be modelled as an instance of CMP and how our algorithm provides an approximation factor of 2 in this case.

In the appendix A we provide a formal analysis of the approximation factor and running time of our algorithm for CMP which is presented in section 3.2. In appendix B we provide the pseudocode for the $4$-approximation algorithm presented in section 3.2. In appendix C we present the pseudocode and run-time analysis of our fully distributed algorithm for \textit{shortcutting}.

\section{The Combinatorial Motion Planning Problem}

In this section we will formally define our problem. But first let's define some preliminaries and introduce the notation we will be using throughout this paper. A path $P$ is defined as a sequence of edges $(v_1, v_2), (v_2, v_3), \dots, (v_{k-1}, v_k)$ such that each edge appears exactly once. The weight of a path $w(P)$ is defined as the sum of the weights of the edges in the paths, i.e., $w(P) = \sum_{i=1}^{k-1} w(v_i,v_{i+1}) $ where $w(i,j)$ denotes the weight of the edge $(i,j)$.

Before we define the Combinatorial Motion Planning Problem, we must define the Multiple Depot-Terminal Hamiltonian Path Problem With Motion Constraints -

\begin{definition}{\textbf{Multiple Depot-Terminal Hamiltonian Path Problem With Motion Constraints (MDTHPPC)}}
Given a weighted, undirected graph $G = (V,E)$ and three subsets of the vertex set $V$ -  a set of depot nodes $D = \{d_1, d_2, \dots, d_k\}$,  a set of terminal nodes $T = \{t_1, t_2, \dots, t_k\}$ and a set of target vertices $Q = V - (D \cup T)$ and also given $k$ traveling salesmen $s_1, s_2, \dots, s_k$ where $s_i$ is initially located at depot $d_i$. Each salesman is assigned a set of motion constraints $\mathcal{C}_i \subseteq Q$ which denotes the set of target vertices that the salesman is allowed to visit. Assign paths to each salesmen $s_i$, starting at $d_i$ and ending at $t_i$ such that the motion constraints are satisfied, each target vertex in $Q$ is visited by exactly one salesman and the sum of weights of paths of all salesmen is minimum.
\end{definition}
\
\newline
We consider a relaxed version of MDTHPPC where the target vertices $q \in Q$ belong to either of the following two types:

\begin{enumerate}
\item Any of the $k$ salesmen can visit the vertex
\item There exists exactly one salesman who is allowed to visit the vertex
\end{enumerate}

In this variant, we can split the set of motion constraints $\mathcal{C}_i$ for each $s_i$ into two disjoint subsets $A_i$ and $F$ where $A_i = \mathcal{C}_i - F$ is the set of vertices that the $i$-th salesman must necessarily visit, i.e., vertices of the second type as described above (since no other salesman can visit this vertex) and $F = \cap_{i=1}^{k} \mathcal{C}_i$ is the common set of vertices that all salesmen are allowed to visit, i.e., vertices of the first type as described above. Thus the vertex set of the graph is given by $V = D \cup T \cup F \cup A$ where $A = \cup_{i=1}^{k} A_i$. Also observe that $Q = F \cup A$. Yadlapalli et al. \cite{3.66} call this variation the \textbf{Combinatorial Motion Planning Problem (CMP)}.
\newline
\newline
Each edge of the graph $(i,j) \in E$ is associated with a real number $w(i,j)$ which denotes the weight of the edge. We also define the weight of a set of edges $S \subseteq E$ of the graph by $c(S) = \sum_{(i,j) \in S} w(i,j)$
\newline
\newline
Sahni and Gonzalez \cite{Sahni} have shown that the traveling Salesman Problem is hard to approximate under general weights. Since CMP is a generalization of TSP, it is also hard to approximate under general weights. Thus we consider the case where the weights of the edges satisfy the triangle inequality, i.e., for any pair of vertices $(i,j)$, $w(i,j) \le w(i,k) + w(k,j) \ \forall \ k \in V$

\section{A Centralized 4-Approximation Algorithm for CMP}

Yadlapalli et al. present a 3.66 factor approximation algorithm for the Combinatorial Motion Planning Problem \cite{3.66}. Henceforth we will refer to it as the YDRP algorithm. Our algorithm has a worse guarantee as compared to the YDRP algorithm for the Combinatorial Motion Planning Problem. In subsection 3.1 we will show that in a fully distributed system, our algorithm will in fact give a better approximation factor than the distributed version of the YDRP algorithm due to lack of exact algorithms for minimum cost perfect matching in distributed environments.

In subsection 3.2 we will present our centralized algorithm for the Combinatorial Motion Planning problem. The complete analysis of the approximation factor and running time can be found in appendix A.

\subsection{Analysis of the Distributed YDRP Algorithm for CMP}
In this section we will analyse what would happen if we try to make the YDRP algorithm distributed and show that it loses its approximation factor.
The YDRP algorithm has an approximation factor of  $\frac{11}{3}$ with a running time of $$\mathcal{O}\left(\sum_{i=1}^{k}|A_i|^3 + |F|^2\pmb\upalpha\left(|F|^2, |F|\right)\right)$$ where $\pmb\upalpha$ is the functional inverse of the Ackermann function. In the worst case this algorithm will run in $\mathcal{O}(n^3)$. Hoogeven's algorithm \cite{Hoogeven} to compute a Hamiltonian path with two specified end points dominates the running time of the YDRP algorithm and this is due to its subroutine which constructs a minimum cost perfect matching for which the current best running time on a graph with $n$ nodes and $m$ edges is $\mathcal{O}(nm)$ as given by Gabow et al. \cite{Gabow}.

The YDRP algorithm constructs a feasible solution for CMP by constructing a set of disjoint paths $\mathsf{HPP_i}$ for each salesman $s_i$ over the vertex set $S_i = \{d_i, t_i\} \cup A_i$ and a set of disjoint cycles $\mathsf{T_{F_i}}$ for each salesman $s_i$ which cover the vertex set $D \cup F$ and then combine $\mathsf{HPP_i}$ and $\mathsf{T_{F_i}}$ for each salesman $s_i$ to produce a path $\mathsf{P_i}$. They then proceed to show that if the optimal solution is represented by $\mathsf{OPT}$ then the following hold true

\begin{equation} \label{eq:1}
\sum_{i=1}^{k} c(\mathsf{HPP_i}) \le \frac{5}{3} \cdot c(\mathsf{OPT})
\end{equation}

$$
\sum_{i=1}^{k} c(\mathsf{T_{F_i}}) \le 2 \cdot c(\mathsf{OPT})
$$

$$
\sum_{i=1}^{k} c(\mathsf{P_i}) \le \sum_{i=1}^{k} c(\mathsf{HPP_i}) + c(\mathsf{T_{F_i}}) \le \frac{11}{3} \cdot c(\mathsf{OPT})
$$

Equation \ref{eq:1} is due to Hoogeven's algorithm. Specifically, Hoogeven shows that for a graph $G = (V,E)$ with a minimum spanning tree $MST$, a minimum cost perfect matching $M$ over all the vertices with odd degree in the spanning tree, an optimal Hamiltonian path $H^*$ between two specified end points of $G$ and the Hamiltonian path $H$ constructed by Hoogeven's algorithm, we have the following

\begin{equation} \label{eq:2}
c(MST) \le c(H^*)
\end{equation}
\begin{equation} \label{eq:3}
c(M) \le \frac{2}{3}\cdot c(H^*)
\end{equation}
\begin{equation} \label{eq:4}
c(H) \le c(M) + c(MST) \le \frac{5}{3}\cdot c(H^*)
\end{equation}

Our primary objective is to come up with a fully distributed algorithm for CMP. If we make a direct attempt to make the YDRP algorithm distributed by using state of the art distributed algorithms for each of the major routines in the YDRP algorithm, then we need a fully distributed version of Hoogeven's algorithm and Hoogeven's algorithm would in turn require a fully distributed algorithm for a minimum cost matching in general graphs. At present the best distributed algorithm for computing a minimum cost matching over a set of vertices is an approximation algorithm with an approximation factor of two \cite{Matching}. If we assume that we are able to create a fully distributed version of YDRP algorithm using this routine, the bound on the cost of the matching with respect to the optimum path $H^*$ would lose a factor of 2 and expressions \ref{eq:3} and \ref{eq:4} would change as follows

$$
c(M) \le 2 \cdot \frac{2}{3}\cdot c(H^*) = \frac{4}{3}\cdot c(H^*)
$$
$$
c(H) \le c(M) + c(MST) \le \frac{7}{3}\cdot c(H^*)
$$

Thus we can see that Hoogeven's algorithm would give an aproximation factor of $\frac{7}{3}$ in distributed systems. If we used this distributed version of Hoogeven's algorithm in a distributed version of the YDRP algorithm, since now we are using a $\frac{7}{3}$-approximation factor algorithm instead of the original $\frac{5}{3}$-approximation algorithm, inequality \ref{eq:1} would change as follows
\begin{equation}
\sum_{i=1}^{k} c(\mathsf{HPP_i}) \le \frac{7}{3} \cdot c(\mathsf{OPT})
\end{equation}
$$
\implies \sum_{i=1}^{k} c(\mathsf{P_i}) \le \sum_{i=1}^{k} c(\mathsf{HPP_i}) + c(\mathsf{T_{F_i}}) \le \frac{13}{3} \cdot c(\mathsf{OPT})
$$

Additionally, at present, we only have distributed algorithms for min cost matching for general graphs in synchronous message passing systems whereas we want to focus in asynchronous message passing systems. 

In the following subsections, we present a $4$-approximation algorithm for CMP which runs in $\mathcal{O}(n^2\alpha(n^2,n))$ time and does not use matching as a subroutine. In section 4, we present a fully distributed version which runs in $\mathcal{O}(n\log{n})$ time and uses at most $\mathcal{O}(n^2)$ messages in asynchronous message passing systems. Thus our algorithm performs better in terms of the approximation factor in distributed systems compared to the YDRP algorithm.

\subsection{A 4-Approximation Algorithm for CMP}

We will solve the problem in three phases, constructing two paths for each salesman, once over the set of vertices $D\cup T \cup A$ (phase I) and once over the set of vertices $D\cup F$ (phase II) and then combine these paths (phase III) to create a set of paths that covers the entire graph $G$. Henceforth, we will use the shorthand $\mathcal{A}$ to refer to this algorithm.

\subsubsection*{Phase I}
Consider the set of vertices $S_i = \{d_i, t_i\} \cup A_i$ for each salesman $s_i$. Since each salesman must necessarily visit all the vertices in $A_i$, we will construct a Hamiltonian path in the set $S_i$ starting at $d_i$ and ending at $t_i$ for each $S_i$.

First, assign the weight of the edge connecting nodes $d_i$ and $t_i$ to 0. Now construct the minimum spanning tree over the subgraph of $G$ containing $S_i$. Call this tree $MST_i$. Double the edges of $MST_i$, i.e., create a multiset of edges where each edge of $MST_i$ appears twice. Now, delete one of the edges connecting $d_i$ to $t_i$. This procedure will give us a multigraph with exactly two vertices with odd degree - $d_i$ and $t_i$.

Now that our graph is Eulerian, we can construct an Euler path from $d_i$ to $t_i$. Call this path $E_i$. Each vertex in the set $S_i$ appears at least once in this path $E_i$.

Now use this path $E_i$ as a guide to obtain a feasible Hamiltonian path over the set of vertices $S_i$ by repeated application of the \textit{shortcut} operation over all repeated occurrences of vertices in the path $E_i$

\begin{definition}{\textbf{Shortcut}}
A shortcut operation consists of replacing two consecutive edges $(i,j), (j,k)$ in a path by a single edge $(i,k)$ and the vertex $j$ is said to have been \textit{shortcutted}. By the triangle inequality it is guaranteed that this new path will have length no greater than the original path.
\end{definition} 

If we perform the shortcut operation over all the duplicate occurrences of vertices in the path $E_i$, eventually we will obtain a path in which each vertex of $S_i$ appears exactly once and thus this \textit{shortcutted} path will be a valid Hamiltonian path over the vertex set $S_i$ which starts at $d_i$ and ends at $t_i$. We will denote this path by $\mathcal{P}_i$ for each set vertex set $S_i$ and the collection of all such $\mathcal{P}_i$ will be our output in phase I. See figure \ref{fig:1}.

\begin{figure}[h!]
  \includegraphics[width=10cm]{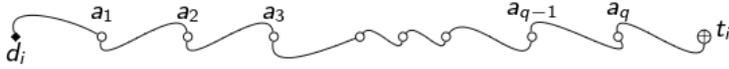}
  \centering
  \caption{Output of Phase I - path $\mathcal{P}_i$ for salesman $s_i$}
  \label{fig:1}
\end{figure}

\subsubsection*{Phase II}
Consider the the subgraph $G'$ of $G$ consisting of the set of vertices $V' = D\cup F$. First, re-assign the weight of the edges connecting each pair of depots to zero. Since these edges have the least weight among all edges of the original graph, it is guaranteed that exactly $k-1$ of them must belong to the minimum spanning tree of this graph. Construct a minimum spanning tree over $G'$

Now, delete the $k-1$ edges which have edge weight zero. Deleting these edges will disconnect all pairs of depots from each other and thus we will obtain a spanning forest $\mathsf{SF}$ consisting of $k$ trees where each tree $T_i \in \mathsf{SF}$ contains exactly one depot, i.e., $d_i$.

Now construct a Hamiltonian cycle over each tree $T_i$ using the well known 2 approximation algorithm for TSP, i.e., double the edges of $T_i$ as we did in phase I, construct an Euler tour over the obtained multigraph and then perform repeated shortcut operations over the obtained Euler Tour on duplicate occurrences of vertices to obtain a cycle where each vertex in $T_i$ appears exactly once, i.e., a Hamiltonian cycle over the vertex set $T_i$. Denote this Hamiltonian cycle for $T_i$ by $\mathcal{C}_i$. The collection of $\mathcal{C}_i$ for each salesman $s_i$ would be the output of phase II. See figure \ref{fig:2}.

\begin{figure}[h!]
  \includegraphics[width=8cm]{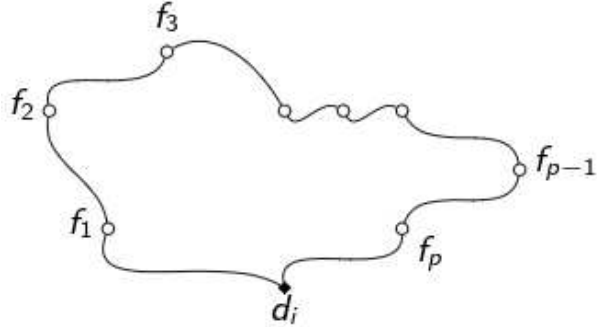}
  \centering
  \caption{Output of Phase II - cycle $\mathcal{C}_i$ for salesman $s_i$}
  \label{fig:2}
\end{figure}

\subsubsection*{Phase III}

In this phase we will combine the paths obtained in the first two phases for each salesman to obtain a feasible solution to CMP

Consider the union of the edges $\mathcal{C}_i$ and $\mathcal{P}_i$ for each salesman $i$. The degree of each depot $d_i$ will be exactly 3 in this union (or will be exactly one if $\mathcal{C}_i = \phi$ and this case can be handled trivially). Thus we can perform a shortcut operation over two of the edges incident to $d_i$ - one from $\mathcal{P}_i$ and the other from $\mathcal{C}_i$ to obtain a path starting at $d_i$ and ending at $t_i$. Denote this path by $\mathcal{H}_i$. Since each vertex in the graph $G$ belongs to exactly one such path $\mathcal{H}_i$, the union of all $\mathcal{H}_i$ is a feasible solution to CMP. See figures \ref{fig:3}, \ref{fig:4} and \ref{fig:5}.

\begin{figure}[h!]
  \includegraphics[width=11cm]{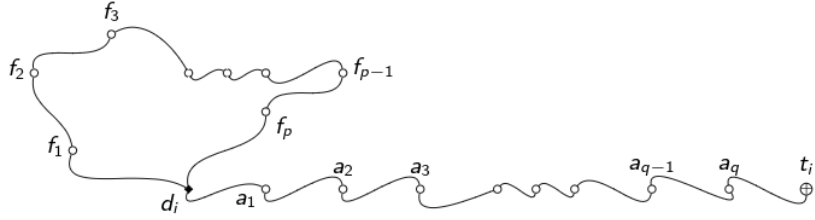}
  \centering
  \caption{Phase III - cycle $\mathcal{C}_i$ and path $\mathcal{P}_i$ are combined}
  \label{fig:3}
\end{figure}

\begin{figure}[h!]
  \includegraphics[width=11cm]{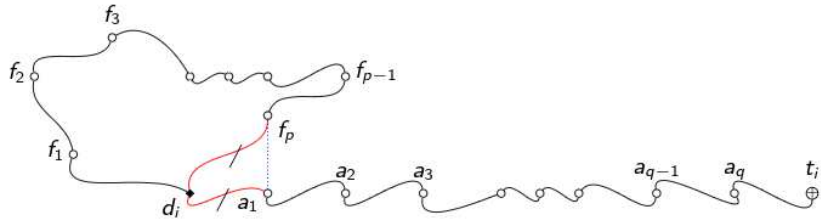}
  \centering
  \caption{Phase III - perform a shortcut over depot $d_i$}
  \label{fig:4}
\end{figure}

\begin{figure}[h!]
  \includegraphics[width=11cm]{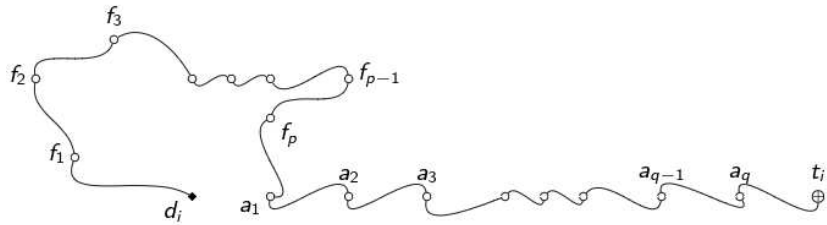}
  \centering
  \caption{Phase III output - path $\mathcal{H}_i$}
  \label{fig:5}
\end{figure}

The complete pseudocode for algorithm $\mathcal{A}$ can be found in appendix B.

\section{A Distributed Approximation Algorithm for CMP}

In this section we present the fully distributed version of the 4-approximation algorithm for CMP which was presented in section 3. Our model of computation is the asynchronous message passing system. We assume that the time taken to pass a message between any two nodes takes at most one unit of time.

Our distributed algorithm is composed of the following three major subroutines -

\begin{itemize}
\item Minimum Spanning Tree - Gallager et al. \cite{GHS} present a fully distributed algorithm to construct the minimum spanning tree of a network with $n$ nodes and $m$ edges in asynchronous message passing systems in $\mathcal{O}(n\log_2{n})$ time and using at most $\mathcal{O}(m + n\log_2{n})$ messages. Henceforth, we will refer to this algorithm as the GHS algorithm.

\item Euler Paths and Euler Tours - S. A. M. Makki et al. \cite{EulerTour} present a fully distributed algorithm to construct the Euler Tour of a network with $n$ nodes and $m$ edges in $\mathcal{O}(m)$ time and using at most $\mathcal{O}(m)$ messages in asynchronous message passing systems.

\item Shortcutting - We present the first fully distributed algorithm in asynchronous message passing systems to perform the shortcutting operation on a path with repeated vertices which works in $\mathcal{O}(n)$ time for a path of length $n$ and uses at most $\mathcal{O}(n)$ messages.
\end{itemize}

\subsection{Distributed Shortcutting}

The complete pseudocode for this algorithm and its run-time analysis can be found in appendix C.

In this section we present our fully distributed algorithm for shortcutting. The input given to our algorithm is a graph $G = (V,E)$ and a path in this graph $\mathcal{P} = (v_1, v_2), (v_2, v_3), \dots, (v_{r-1}, v_r)$. Each processor in the distributed system corresponds to a vertex $v \in V$. Note that if $\exists \ v \in V$ such that $v \not\in \mathcal{P}$ then deleting $v$ from the vertex set will have no effect on the output of the algorithm. Thus we will assume that the path $\mathcal{P}$ is a Hamiltonian path over the graph.

Each vertex $v \in V$ maintains two adjacency lists -
\begin{itemize}
\item $adj_g$ - the adjacency list of this vertex as it is defined by the set of edges $E$
\item $adj_\mathcal{P}$ - the adjacency list of this vertex as is defined by the path $\mathcal{P}$, i.e., a list of edges that are incident to $v$ in $\mathcal{P}$. Each edge in this list is also assigned an integer label according to the order in which these edges appear in $\mathcal{P}$. The edge that appears before any other edge in this list is given the label $1$, the edge that appears next is given the label $2$ and so on. This is similar to how a path is represented in the distributed Euler Tour algorithm \cite{EulerTour}
\end{itemize}
The algorithm modifies the list $adj_{\mathcal{P}}$ for each vertex $v \in V$ to produce another path $\mathcal{P}' = (v_{j_1}, v_{j_2}), (v_{j_2}, v_{j_3}),$ $ \dots, (v_{j_{l-1}}, v_{j_l})$ such that $v_{j_x} \not= v_{j_y} \ \forall \ j_x, j_y$.

Using this representation, one can easily traverse the path $\mathcal{P}$ as follows - each node in the network maintains a variable $cnt$ which denotes the number of edges from the start of the list $adj_\mathcal{P}$ that have already been traversed by us. Initially $cnt = 0$. Each time we visit a node $v_i$ in our traversal via some edge that belongs to the path (or if $v_i$ is the starting vertex of the path), we increment the counter variable $cnt$ by 1 for $v_i$. If $v_i$ is not the last node in the path $\mathcal{P}$ then there must exist an edge in the adjacency list $adj_{\mathcal{P}}$ for $v_i$ whose label is $cnt+1$ and this is the next edge to use in the traversal of path $\mathcal{P}$. Thus we increment $cnt$ by 1 and visit the neighbour of node $v_i$ across the edge with label $cnt$. If there is no edge with label $cnt+1$ then $v_i$ is an endpoint of the path and we have complete our traversal. This procedure can trivially be modified in asynchronous message passing systems where the ``traversal" across an edge $(v_i, v_j)$ is simply simulated by node $v_i$ sending a message across the edge to $v_j$

We can modify this basic traversal algorithm to perform shortcuts over vertices that appear more than once by simply introducing a single boolean variable $vis$ for each vertex $v \in V$. Initially $vis$ is set to $false$ for all nodes except the terminal node of the path $\mathcal{P}$ for which we set $vis$ to $true$ (this is done to allow us to shortcut intermediate occurrences of the terminal node. This can be easily modified if the path is actually a cycle to avoid performing shortcut on the first vertex of the path).

Each time an unvisited node is visited for the first time during the traversal, it sets its variable $vis$ to $true$ and continues the traversal algorithm as before. If on the other hand we visit a vertex which has already been visited before, we must perform a shortcut operation over this vertex. Let $v$ be the vertex over which we must perform the shortcut operation. Let $u$ be the vertex from which we arrived at $v$ in the path and let $w$ be the vertex that we would have visited from $v$ if we did not perform the shortcut operation. Now, we must replace the sequence of edges $(u,v), (v,w)$ with a single edge $(u,w)$. 

To facilitate the shortcut, node $v$ sends a message to node $u$ instructing it to replace the edge $(u,v)$ with  $(u,w)$ and similarly $v$ sends a message to $w$ instructing it to replace the edge $(v,w)$ with $(u,w)$. After this, $v$ deletes the edges $(u,v)$ and  $(v,w)$ form $adj_\mathcal{P}$ and increments $cnt$ by 1 as it would have in case it traversed along the edge $(v,w)$. Eventually, we will reach the final node in the path and at this point our algorithm will terminate and the collection of $adj_\mathcal{P}$ for each vertex will define the new \textit{shortcutted} path $\mathcal{P}'$.

Henceforth, we will refer to this algorithm with the shorthand $\mathcal{B}$.

\begin{lemma}
The fully distributed algorithm $\mathcal{B}$ for repeated shortcutting eventually terminates.
\end{lemma}
\begin{proof}
Since the algorithm traverses across each edge exactly  once due to variable $cnt$ in each vertex $v$, eventually the algorithm will traverse across the final edge of the path and thus will terminate.
\end{proof}

\begin{lemma}
The fully distributed algorithm $\mathcal{B}$ for repeated shortcutting computes the shorcutted path correctly.
\end{lemma}
\begin{proof}
\
\begin{itemize}
\item If a vertex $v$ is not the terminal vertex of the path $\mathcal{P}$ then it is easy to see that all its occurrences after its first occurrence will be shortcutted. This is because one the variable $vis$ is toggled to be $true$ it remains $true$.
\item If a vertex $v$ is the terminal vertex of the path then the repeated occurrences of $v$ in $\mathcal{P}$ that require shortcutting are all occurrences except its final occurrence. To allow us to detect start shortcutting of the repeated occurrences immediately, we set $vis = true$ for the terminal vertex.
\end{itemize}
\end{proof}

\begin{theorem}
The fully distributed algorithm $\mathcal{B}$ for repeated shortcutting is correct.
\end{theorem}
This follows from lemmas 4.1 and 4.2

\subsection{Distributed CMP}

We are now ready to create a fully distributed algorithm for CMP. The input to the algorithm is an undirected weighted graph $G = (V,E)$ where the weight function over the edges satisfies the triangle inequality. Additionally, we are also given the set of depots $D$, the set of terminal nodes $T$, the set of vertices $A_i$ for each salesman $s_i$ and the set of common vertices $F$. Each node has information about which category it belongs to, i.e., if it is a depot or a terminal node or a target vertex and in case it is a target vertex, which set of vertices - $A_1, A_2, \dots, A_k$ or $F$ it belongs to. Each node $v \in V$ maintains an adjacency list $adj_v$ of its edges.

Initially, when a node wakes up, it sends a message across each of its edges, requesting them to report which category they belong to. Once this step is completed, each  pair of $d_i$ and $t_i$ set the weight of the edge connecting $d_i$ to $t_i$ to zero in their respective adjacency lists. Following this, all nodes $v \in S_i = \{d_i, t_i\} \cup A_i$ initiate the GHS algorithm to construct a minimum spanning tree over $S_i$. Once the minimum spanning tree is constructed, the ``core'' edge of the final fragment of the GHS algorithm sends a message to all other nodes in the tree to double their edges except for the edge connecting $d_i$ to $t_i$. To avoid confusion between the duplicate edges we assume that the edges are assigned distinct ids which can be done in a trivial fashion by adding a large enough number to the ids of the original edges. After this, the core edge also sends a message to the unique depot $d_i \in S_i$ to initiate the distributed Euler Path algorithm  of S. A. M. Makki et al. \cite{EulerTour} ($d_i$ already knows the unique terminal node $t_i$ of this set). Once the Euler Path algorithm is completed, the unique labellings given to the edges of the Euler Path by each vertex $v \in S_i$ will form a path $\mathcal{P}_i$ which subsequently acts as an input for our distributed shortcutting algorithm as described earlier. Application of the distributed shortcutting algorithm on $\mathcal{P}$ will produce the Hamiltonian Path $\mathcal{P}_i$ for the vertex set $S_i$. This ends phase I for the distributed algorithm.

Upon completion of phase I, each depot sends a message to all other depots to indicate that it is ready to start phase II of the algorithm. Once all the depots have completed phase I, they re-assign the weights of the edges connecting them to other depots to 0 and then the set of vertices $D \cup F$ initiates the GHS algorithm to construct a minimum spanning tree over $D \cup F$. Once the MST is constructed the ``core'' edge of the MST passes a message to each depot instructing them to delete the edges of the MST that connect them to other depots. This results in the $k$-spanning forest $SF$. Each depot then broadcasts a message over its own tree of the spanning forest, instructing all nodes to double their edges and then initiates the Euler Tour algorithm followed by shortcutting, similar to what was done in phase I for each $S_i$. This will result in a Hamiltonian cycle $\mathcal{C}_i$ for each tree $T_i \in SF$. After this, each depot $d_i$ will have two adjacency lists which are the outputs of the shortcutting algorithm in phases I and II respectively - $adj_{\mathcal{P}_i}$ for $\mathcal{P}_i$ and $adj_{\mathcal{C}_i}$ for $\mathcal{C}_i$. These depots increase the label of all edges in $adj_{\mathcal{P}_i}$ by $|adj_{\mathcal{C}_i}|-1$ and then shortcut the last edge in $\mathcal{C}_i$ and first edge of $\mathcal{P}_i$ to produce a path $\mathcal{H}_i$ for the salesman $s_i$. This would result in the end of phase III and completes the distributed algorithm.

\subsubsection{Complexity Analysis}

Initially $\mathcal{O}(n)$ time is spend by the nodes to identify which sets their neighbours belong to and this uses $\mathcal{O}(n^2)$ messages. The running time and message complexity in phase I is dominated by the GHS algorithm which takes $\mathcal{O}(|S_i|\log_2{|S_i|})$ time and $\mathcal{O}(|S_i|^2)$ messages for each set $S_i$. The depots then use $\mathcal{O}(k^2)$ messages to indicate to each other that they have finished phase I. Following this, phase II begins and its runtime and message complexity is also dominated by the GHS algorithm taking $\mathcal{O}(|F\cup D|\log_2{|F\cup D|})$ time and $\mathcal{O}(|F\cup D|^2)$ messages. Phase III takes constant time and constant number of messages for each depot $d_i$.

Thus the final run time of the algorithm is given by $$\mathcal{O}\left(n+ \max_{1 \le i \le k} |S_i|\log_2{|S_i|} + |F\cup D|\log_2{|F\cup D|} + k\right)$$ and uses at most $$\mathcal{O}\left(n^2 + \sum_{i=1}^{k} |S_i|^2 + |F\cup D|^2 + k^2\right)$$ messages. Both of these quantities are upper-bounded by $\mathcal{O}(n\log_2{n})$ and $\mathcal{O}(n^2)$ respectively where $n = |V|$.

\section{2-Approximation Algorithm for the Multiple Traveling Salesman Problem}


In this section we demonstrate how the 4-approximation algorithm for CMP can be used to solve the Multiple Traveling Salesman Problem and then demonstrate that we obtain an approximation factor of 2 in this special case.

\begin{definition}{\textbf{$\mathbf{k}$-Traveling Salesman Problem}}
Given a weighted undirected graph $G = (V, E)$ and $k$ traveling salesmen $s_1, s_2, \dots, s_k$ located at their depots given by $D$ = $\{d_1, d_2, \dots, d_k\}$, find a cycle for each salesman such that each city is visited by exactly one salesman and the sum of weights of all paths is minimized. It is also known as the $k$-TSP problem.
\end{definition}
\
\newline
It is easy to observe from the definition that CMP is just a generalized version of $k$-TSP. If $d_i = t_i$ and $A_i = \phi \ \forall i$ then CMP reduces to an instance of $k$-TSP and thus we can use our algorithm $\mathcal{A}$ from Section 3.2 to solve $k$-TSP also.
\newline
\newline
We can show that $\mathcal{A}$ gives an approximation factor of $2$ in $k$-TSP. Since $A_i = \phi \ \forall i$, we can simply skip Phase I of $\mathcal{A}$. This would shave off a cost of $2\cdot c(\mathsf{OPT})$ from our approximation factor. Giving an approximation factor of $2$. Since there is no other modification to $\mathcal{A}$, our runtime analysis, message complexity and proof of correctness for both the centralized and distributed algorithms remain the same for $k$-TSP also.

\section{Conclusion}

We presented the first fully distributed approximation algorithm for the Combinatorial Motion Planning 
problem in asynchronous message-passing systems, which runs in $\mathcal{O}(n\log_2{n})$ time and uses 
at most $\mathcal{O}(n^2)$ messages; and the first fully distributed approximation algorithm for the  
multiple Traveling Salesman problem ($k$-TSP) which runs in $\mathcal{O}(n\log_2{n})$ time and uses 
at most $\mathcal{O}(n^2)$ messages, which are tight bounds for MSTs~\cite{lowerbound}.

The fully distributed algorithm for shortcutting should encourage more research for fully distributed 
approximations for other variants of the vehicle routing problem, which as we mentioned earlier has a 
vast number of practical applications.

Further work is needed in fully distributed exact algorithms for minimum-cost matching in general 
graphs.  An exact algorithm would allow us to directly create a distributed version of the YDRP 
algorithm, thus giving us a $\frac{11}{3}$-approximation algorithm in distributed systems as well.

Our $2$-approximation algorithm for $k$-TSP is a result of a direct reduction of $k$-TSP to an 
instance of CMP, does not necessarily exploit well any of the  underlying structure of $k$-TSP 
itself.  This leaves open the possibility that if we take advantage of some property of $k$-TSP 
itself, it might be possible to create a distributed approximation algorithm for $k$-TSP with a better approximation factor.


\begin{appendices}

\section{Analysis of Approximation Factor of the 4-Approximation Algorithm}

\begin{definition}{\textbf{Approximation Algorithm}}
An approximation algorithm $\mathcal{X}$ is an algorithm that finds approximate answers to optimization problems. For minimization problems, given an instance $I$ for some problem $P$ such that the optimum value for this instance with respect to the problem $P$ is $OPT(I)$, the algorithm $\mathcal{X}$ produces a feasible solution with cost $\mathcal{X}(I)$ such that $\mathcal{X}(I) \le \alpha OPT(I) \ \forall \ I$ where $\alpha \ge 1$ is called the approximation factor.
\end{definition}

\begin{theorem}
Algorithm $\mathcal{A}$ from section 3.2 is a $4$-approximation algorithm for CMP
\end{theorem}

\begin{proof}

Denote the optimum solution for CMP by $\mathsf{OPT}$. This consists of $k$ edge disjoint paths $\mathsf{OPT_i}$, one path for each salesman $s_i$.

First, bound the weight of the path $\mathcal{P}_i$ for each salesman $s_i$. Let the optimum Hamiltonian path over the vertex set $S_i$ starting at $d_i$ and ending at $t_i$ be $H^*_i$. $\mathsf{OPT_i}$ consists of a path starting at $d_i$ and ending at $t_i$ and consists of all vertices in $A_i$ and possibly some vertices from the set $F$. If we perform repeated shortcut operations over the vertices from the set $F$ from each path $\mathsf{OPT_i}$, then we would obtain a valid Hamiltonian path starting at $d_i$ and ending at $t_i$ over the vertex set $S_i$. Since the optimal path of this kind is $H^*_i$, we trivially have that

$$c(H^*_i) \le c(\mathsf{OPT_i})$$.

Since the path $\mathcal{P}_i$ was obtained by doubling the edges of the minimum spanning tree $MST_i$ followed by shortcutting, and because the weight of an minimum spanning tree - $c(MST_i)$ is a known lower bound for $c(H^*_i)$, we have that

\begin{equation} \label{eq:11}
c(\mathcal{P}_i) \le 2 \cdot c(MST_i) \le 2 \cdot c(H^*_i) \le 2 \cdot c(\mathsf{OPT_i})
\end{equation}
\begin{equation} \label{eq:12}
\implies \sum_{i=1}^{k} c(\mathcal{P}_i) \le 2 \cdot \sum_{i=1}^{k}c(\mathsf{OPT_i}) = 2 \cdot c(\mathsf{OPT})
\end{equation}

Now we will bound the quantity $\sum_{i=1}^{k} c(\mathcal{C}_i)$ in terms of $c(\mathsf{OPT})$. If we shortcut the vertices belonging to the set $A = \cup_{i=1}^{k} A_i$ from $\mathsf{OPT}$ we obtain a spanning forest over the set of vertices $D \cup F$ which consists of exactly $k$ trees and each tree consists of exactly one depot. Since we constructed the spanning forest $\mathsf{SF} = \cup_{i=1}^{k}T_i$ via a minimum spanning tree, we have that 

$$\sum_{i=1}^{k} c(T_i) \le c(\mathsf{OPT})$$

But we know that due to our 2-approx algorithm for construction of a least weight Hamiltonian cycle over the set of vertices defined by each tree $T_i$, we have $c(\mathcal{C}_i) \le 2 \cdot c(T_i) \ \forall \ i$. This implies the following

\begin{equation} \label{eq:13}
\sum_{i=1}^{k}c(\mathcal{C}_i) \le 2\cdot \sum_{i=1}^{k} c(T_i) \le 2 \cdot c(\mathsf{OPT})
\end{equation}

Each path $\mathcal{H}_i$ is obtained by combining $\mathcal{C}_i$ and $\mathcal{P}_i$ via one shortcut operation over the depot $d_i$. By the triangle inequality and equations \ref{eq:12} and \ref{eq:13}, we get

\begin{equation}
\sum_{i=1}^{k} c(\mathcal{H}_i) \le \sum_{i=1}^{k} c(\mathcal{P}_i) + c(\mathcal{C}_i) \le 2 \cdot c(\mathsf{OPT}) + 2 \cdot c(\mathsf{OPT}) = 4 \cdot c(\mathsf{OPT})
\end{equation}

Thus our algorithm $\mathcal{A}$ gives us an approximation factor of 4.

\end{proof}

\subsection{Running Time Analysis}

Our algorithm employs three main routines - computation of minimum spanning trees, computation of Euler tours and the shortcut operations.

\begin{itemize}
\item Construction of a minimum spanning tree on a graph with $n$ nodes and $m$ edges can be done in $\mathcal{O}(m\pmb\upalpha(m,n))$ time where $\pmb\upalpha$ is the functional inverse of the rapidly growing Ackermann function \cite{FastMST}

\item Construction of an Euler Tour on a graph with $n$ nodes and $m$ edges can be done in $\mathcal{O}(n+m)$ time \cite{CentralizedEuler}

\item Shortcutting the repeated vertices of a path of length $n$ can be done trivially in $\mathcal{O}(n)$ time with the aid of doubly linked lists and keeping track of visited vertices in the path.

\end{itemize}
\
\newline
Using this we can analyze the run time of $\mathcal{A}$ as follows

\begin{enumerate}
\item \textit{Running Time in Phase I}: For each subgraph of $G$ defined by the vertex set $S_i$, the construction of a minimum spanning tree takes $\mathcal{O}\left(|A_i|^2\pmb\upalpha(|A_i|^2,|A_i|)\right)$ time. Doubling the edges and computing an Euler path over this multigraph will take $O(|A_i|)$ time. Since the Euler path will also be of length $\mathcal{O}(|A_i|)$, the shortcutting step will also take $\mathcal{O}(|A_i|)$ time. Thus phase I is dominated by the computation of a minimum spanning tree. Summing up over all sets $S_i$, phase I takes $\mathcal{O}(\sum_{i=1}^{k}|A_i|^2\pmb\upalpha(|A_i|^2, |A_i|))$ time.

\item \textit{Running Time in Phase II}: Re-assigning weights to edges connecting the depots will take $\mathcal{O}(k^2)$ time. We construct the MST just once over the set of vertices $D\cup F$. This will take $\mathcal{O}((|F|+k)^2\pmb\upalpha((|F|+k)^2), (|F|+k))$ time and this will dominate the running time in phase II similar to how it did in phase I since all other steps are done in linear time in the number of nodes.

\item \textit{Running Time in Phase III}: Since phase III involves performing a simple shortcut operation for each salesman $i$ which takes $\mathcal{O}(|\mathcal{C}_i| + |\mathcal{P}_i|)$ time per salesman, summing over all salesman $s_i$, we get a running time of $\mathcal{O}(|V|)$

\item \textit{Total Running Time}: Adding up the running times for each phase, we get a total running time of $$\mathcal{O}\left(\sum_{i=1}^{k} |A_i|^2\pmb\upalpha\left(|A_i|^2, |A_i|\right) + (|F|+k)^2\pmb\upalpha\left((|F|+k)^2, |F|+k\right) + |V|\right)$$

\item \textit{Worst Case Running Time}: In the worst case, the running time will be dominated by the running time for the cost of construction of minimum spanning tree.  If $n = |V|$ our running time can also be neatly written as $\mathcal{O}(n^2\pmb\upalpha(n^2, n))$
\end{enumerate}

\section{Pseudocode for Centralized 4-Approximation Algorithm for CMP}

Given below is the pseudocode for the centralized 4-approximation algorithm of CMP. It makes use of four other routines - MinSpanningTree(), EulerPath(), EulerTour() and Shortcut(), which return the minimum spanning tree of a given graph \cite{FastMST}, the Euler Tour of a given graph \cite{CentralizedEuler}, the Euler path in a given graph between two specified endpoints \cite{CentralizedEuler} and the repeated shortcut operation as explained in the algorithm above respectively.

\begin{algorithm}[H]
\caption{Pseudocode for $\mathcal{A}$}
\begin{algorithmic}[1]
\Procedure{CMP}{$D$, $T$, $A$, $F$, $E$, $k$}
\For{i = $1$ to $k$}
	\State $w_{d_i, t_i} = 0$
	\State $S_i = \{d_i, t_i\} \cup A_i$
	\State $MST_i =$ MinSpanningTree$(S_i, E)$
	\State $MST_i = MST_i + MST_i$
	\Comment{Double the edges}
	
	\State $MST_i = MST_i / (d_i, t_i)$
	\Comment{Delete one of the edges connecting $d_i$ to $t_i$}	
	
	\State $E_i =$ EulerPath$(MST_i, d_i, t_i)$
	\State $\mathcal{P}_i = $ ShortCut$(E_i)$
\EndFor
\newline
\For{i = $1$ to $k$}
	\For{j = i+1 to $k$}
		\State $w_{ij} = 0$
		\Comment{Re-assign weights of edges connecting depots}
	\EndFor
\EndFor
\newline
\State $SF$ = MinSpanningTree$(D\cup F, E)$
\State $SF = SF / \{(i,j) \ | \ w_{ij} = 0\}$
\Comment{Delete the 0 weight edges}
\newline
\For{i = 1 to $k$}
	\Comment{Construct a Hamiltonian cycle over each tree of $SF$}
	\State $\mathcal{C}_i =$ Shortcut$($EulerTour$(T'_i + T'_i)$ $)$
\EndFor
\newline
\For{i = 1 to $k$}
	\Comment{Combine $\mathcal{P}_i$ and $\mathcal{C}_i$ for each salesman}
	\State $\mathcal{H}_i =$ Shortcut$(\mathcal{P}_i + \mathcal{C}_i)$
\EndFor

\State \Return{$\mathcal{H} = \cup_{i=1}^{k} \mathcal{H}_i$}

\EndProcedure
\end{algorithmic}
\end{algorithm}

Lines 2 to 9 correspond to phase I of the algorithm where in we construct a minimum spanning tree over each set of vertices $S_i$, double its edges, delete one copy of the edge connecting $d_i$ to $t_i$ and then construct a Euler Tour followed by a shortcutting step to give us a Hamiltonian Path starting at $d_i$ and ending at $t_i$ for each vertex set $s_i$.

Lines 11 to 20 correspond to phase II of the algorithm where we first re-assign the weights of the edges connecting all depots to 0, construct a minimum spanning tree and subsequently a spanning forest by deleting all 0 weight edges from this forest such that each tree of the forest contains exactly one depot. Then we proceed to construct a Hamiltonian cycle over the graph for each vertex set defined by a tree of the spanning forest.

Lines 21 to 23 correspond to phase III of the algorithm. Each path $\mathcal{H}_i$ is obtained by combining the corresponding path $\mathcal{P}_i$ and cycle $\mathcal{C}_i$ via a single shortcut operation over the depot $d_i$. Since the depot $d_i$ is the only vertex that is repeated in the path $\mathcal{P}_i + \mathcal{C}_i$, we can simply write this single shortcut operation as Shortcut$(\mathcal{P}_i +\mathcal{C}_i)$

In Line 24, we collect all the paths $\mathcal{H}_i$ for each salesman $s_i$ and this collection of routes $\mathcal{H}$ is the output of our algorithm.

\section{Pseudocode for Distributed Shortcutting and Complexity Analysis}

Algorithm 2 gives the fully distributed pseudocode for shortcutting for a path $\mathcal{P} = (v_1, v_2), (v_2, v_3), \dots, (v_{r-1}, v_r)$. This code is executed for each process $p_i$ in the network. There also exist two distinguished processes (one in case we have a cycle instead of the path), the end points of the path.

In line 1, variables $cnt$, $vis$, $adj_\mathcal{P}$ and $adj_g$ correspond to the variables as described in section 4.1. An additional boolean variable $start$ is used to indicate whether a node has woken up or not.

In lines 3-9, when a processor wakes up spontaneously, it checks if it is the terminal node of the path and if yes, sets its variable $vis$ to $true$. If the processor corresponds to node $v_1$ then it makes a call to the moveForward() sub-routine.

Lines 30 to 34 correspond to the moveForward() sub-routine which corresponds to identifying the next node that must be visited in the traversal of the path and passing a message to it to indicate that the node adjacent along this edge must continue the algorithm.

In Lines 10 to 18, when a processor receives a node, if it is the last node in the path, it terminates. Else it identifies if it has been visited before in the traversal of the path or not as per value of the variable $vis$. If this is the first time a message has arrived at said processor, then it sets variable $vis$ to true and calls the moveForward() sub-routine else it performs a shortcut along itself.

The shortcut operation is performed with the aid of the shortcut sub-routine in lines 35 to 40 which identifies which two edges adjacent node must be shortcutted. The message $\langle R, x, y, id \rangle$ instructs the node receiving this message to replace the edge going towards $x$ (or incoming from $x$) with a an edge going towards $y$.

It is not hard to see that the algorithm sends at most two messages over each edge of the path $\mathcal{P}$ (at most once per moveForward() and at most once for a shortcut() operation()) and hence it runs in $\mathcal{O}(|\mathcal{P}|)$ time. This also gives a bound of $\mathcal{O}(|\mathcal{P}|)$ on the message complexity of the algorithm.

\begin{algorithm}[H]
\caption{Distributed Shortcutting, code for processor $p_i$}
\begin{algorithmic}[1]
\State $cnt = 0,\ vis = false,\ adj_\mathcal{P},\ adj_g,\ start = false$
\State upon receiving no message:
\Indent
	\If{$start = false$}
		\State $start = true	$
		\If{$p_i = v_r$}
			$vis = true$
		\EndIf
		\If{$p_i = v_1$}
			moveForward()
		\EndIf
	\EndIf
\EndIndent
\newline

\State upon receiving $\left\langle M \right\rangle$ from $p_j$:
\Indent
	\State $cnt = cnt + 1$
	\If{$cnt = |adj_\mathcal{P}|$ }
		\textbf{terminate}
	\EndIf	
	
	\If{$vis = false$}
		\State moveForward()
	\Else
		\State shortcut($p_j$)
	\EndIf
\EndIndent
\newline

\State upon receiving  $\left\langle R, x, y, id \right\rangle$ from $p_j$
\Indent 
	\State $b$ = label of $x$ in $adj_\mathcal{P}$
	\State $adj_\mathcal{P}(b) = y$
	\If{$id = 1$}
		\State $cnt = cnt + 1$
		\If{$vis = false$}
			\State moveForward()		
		\Else
			\State shortcut()
		\EndIf
	\EndIf
\EndIndent
\newline

\Procedure{moveForward}{}()
	\State $vis = true$
	\State $cnt = cnt + 1$
	\State send $\left\langle M \right\rangle$ across $adj_\mathcal{P}(cnt)$
\EndProcedure
\newline

\Procedure{shortcut}{$p_j$}
	\State $p_k = adj_\mathcal{P}(cnt+1)$
	\State send $\left\langle R, p_i, p_k, 0 \right\rangle$ to $p_j$
	\State send $\left\langle R, p_i, p_j, 1 \right\rangle$ to $p_k$
	\State $cnt = cnt + 1$
\EndProcedure
\end{algorithmic}
\end{algorithm}

\end{appendices}

\end{document}